\documentclass[12pt,reqno,fleqn]{amsart}
\usepackage{latexsym}
\usepackage{amssymb}
\usepackage{amsxtra}
\usepackage{amsmath}
\usepackage{amsfonts}
\newtheorem{theorem}{Theorem}[section]
\newtheorem{lemma}[theorem]{Lemma}
\newtheorem{proposition}[theorem]{Proposition}

\usepackage{amssymb}
\usepackage{mathrsfs}
\usepackage{amsmath}
\usepackage{enumerate}
\def\[{\begin{eqnarray}}
\def\]{\end{eqnarray}}
\renewcommand{\L}{\mathcal{L}}

\newcommand{\M}{\mathcal{M}}
\newcommand{\Z}{\mathbb{Z}}
\renewcommand{\d}{\partial}
\newcommand{\N}{\mathbb{N}}
\makeatother       

\makeatletter      
\@addtoreset{equation}{section}
\makeatother       

\begin{document}

\title[constrained discrete multi-KP]{Virasoro  symmetry of the  constrained multi-component KP hierarchy and its integrable discretization}
\author{
Chuanzhong Li\dag,\  \ Jingsong He\ddag}
\allowdisplaybreaks
\dedicatory {\small Department of Mathematics,  Ningbo University, Ningbo, 315211, China\\
\dag email:lichuanzhong@nbu.edu.cn\\
\ddag email:hejingsong@nbu.edu.cn}
\thanks{}
\date{}

\begin{abstract}
In this paper, we construct the Virasoro type additional symmetries of a kind of
constrained multi-component KP  hierarchy and give the Virasoro flow equation on eigenfunctions and adjoint eigenfunctions. It can also be seen that the
algebraic structure of  the Virasoro symmetry is kept after discretization from the constrained multi-component KP hierarchy to the discrete constrained multi-component KP hierarchy.
\end{abstract}


\maketitle

\tableofcontents

\allowdisplaybreaks
\section{ Introduction}

The KP hierarchy is one of the most important nonlinear integrable system in mathematics and physics. The additional symmetry of the KP hierarchy was
constructed by L. A. Dickey in \cite{dicheymdern1} in which the additional symmetry constitutes a famous algebra called Virasoro algebra. The KP system becomes one of the most fundamental integrable models in mathematical physics as Toda system.
As one of the most important sub-hierarchies of the KP hierarchy by considering different
reduction conditions on the Lax operator, the
constrained KP hierarchy contains a large number of interesting
 soliton equations under the
so-called symmetry constraint\cite{kss,chengyiliyishenpla1991,
chengyiliyishenjmp1995} which means that the negative part of the Lax operator of the
constrained KP hierarchy
is a generator of the additional symmetries of the KP hierarchy.
 However, the
additional symmetry flows of the KP hierarchy are not consistent with the
special form of the  Lax operator for the constrained KP hierarchy.
Therefore it is highly non-trivial  to find  a suitable  additional
symmetry  flows on eigenfunctions and adjoint eigenfunctions for this sub-hierarchy\cite{aratyn1997} to make the constraint compatible with its additional symmetry.
In addition to the above-mentioned
 integrable systems, other KP type and Toda type systems  also have interesting  algebraic structure
in additional symmetries\cite{ourBlock,dtyptds,ghostdKP,maohuajmp}.

The multi-component KP hierarchy is important matrix-formed generalization of the original KP hierarchy and its additional symmetry was well-studied in \cite{dicheymdern}. In \cite{kacvandeleur},  the representation related to the multi-component KP hierarchy was proposed. In a certain sense, the multi-component KP hierarchy served as a universal integrable hierarchy because many well known integrable hierarchies come from certain reductions of it, for example, like the Davey-Stewartson system, the two-dimensional Toda lattice, the $n$-wave
system and the Darboux-Egoroff system.
The Hirota bilinear equations of the multi-component KP hierarchy play an interesting role in the theory
of theta functions of Riemann surfaces. To be more specific, it shows that the theta function
of a Riemann surfaces gives a tau function for the $n$-component KP hierarchy in \cite{Maffei}. It also show that such a theta function satisfies the Fay trisecant formula.
There is a new interesting phenomenon in the $n$-component case which does not occur in the single-component
case: the tau function and the wave function are a collection of functions
parameterized by the elements of the root lattice of type A.  Later with reduction, the special but quite natural
constrained multi-component KP hierarchy was studied by Y. Zhang in \cite{zhangJPA}. But for the additional symmetry of this reduced hierarchy it is still in blank and interesting. That becomes one motivation of us to do this work. Recently the research on random matrices and non-intersecting Brownian motions and the study of moment matrices with regard to several weights showed that the determinants of such moment matrices these
determinants are tau-functions of the multi-component KP-hierarchy \cite{randommatrixCMP}. Basing on above mentioned importance of the multi-component KP-hierarchy, it is quite necessary to go on further study on this hierarchy and its reduced hierarchy.

Besides above continuous dynamical systems, discrete integrable systems becomes more and more important. For multicomponent discrete integrable systems, the multicomponent Toda hierarchy recently attracts a lot of valuable research such as \cite{manas,EMTH,EZTH}.
The Hamiltonian structures and
 tau function for the discrete KP hierarchy
was introduced in (\cite{Kupershimidt}-\cite{Iliev}). Later the additional symmtries and ghost symmetry for the discrete KP hierarchy was considered in \cite{ghostdKP,LiuS2} respectively.
   In our recent work \cite{maohuajmp}
 the additional symmetries of one-component constrained
discrete KP hierarchy was constructed in which a suitable  additional
symmetry  flows on eigenfunctions and adjoint eigenfunctions for this hierarchy was given.  In this paper, we also consider the discrete constrained multi-component KP hierarchy and identify its algebraic structure hidden behind this discrete hierarchy.

Our main purpose of this article is to identify the Virasoro symmetry of the constrained multi-component KP hierarchy hierarchy and give the complete Virasoro
flows on eigenfunctions and adjoint eigenfunctions in the Lax operator of the hierarchy which form the positive half of Virasoro algebra using the technique in \cite{aratyn1997,maohuajmp}. Also we give the Virasoro symmetry of the discrete constrained multi-component KP hierarchy. That shows  that the
algebraic structure of  the Virasoro symmetry is kept after discretization from the constrained multi-component KP hierarchy to the discrete constrained multi-component KP hierarchy.

This paper is organized as follows. We give a brief description of the
constrained $N$-component  KP hierarchy in Section 2. The main results are stated and proved in Section 3,
which concerns the Virasoro symmetry of the multi-component constrained KP hierarchy.  Section 4 and Section 5 will concerns the Virasoro symmetry of the discrete multi-component constrained KP hierarchy. Section 6 is devoted
to conclusions and discussions.



\section{The constrained multi-component KP hierarchy}

For an $N$-component KP hierarchy, there are $N$ infinite families of time variables $t_{\alpha n}, \alpha=1,\ldots,N, n=1,2,\ldots$. The coefficients $u_1, u_2,\ldots$ of the Lax operator
 \[L=E\partial +u_1\partial^{-1}+u_2\partial^{-2}+\ldots\]
 are $N\times N$ matrices and $\d$ is the derivative about spatial variable $x$.  There are another $N$ pseudo-differential operators $R_1,\ldots, R_N$ in the matrix form
$$R_{\alpha}=E_{\alpha}+u_{\alpha 1}\partial^{-1}+u_{\alpha 2}\partial^{-2}+\ldots,$$ where $E_{\alpha}$ is the $N\times N$ matrix with $1$ on the $(\alpha,\alpha)$-component and zero elsewhere, and $u_{\alpha 1}, u_{\alpha 2},\ldots$ are also $N\times N$ matrices. The operators $L, R_1,\ldots, R_N$ satisfy the following conditions:
$$LR_{\alpha}=R_{\alpha} L, \quad R_{\alpha}R_{\beta}=\delta_{\alpha\beta}R_{\alpha}, \quad \sum_{\alpha=1}^N R_{\alpha}= E.$$

 The Lax equations of the $N$-component KP hierarchy are:
\begin{equation*}
\frac{\partial L}{\partial t_{\alpha, n}}=[B_{\alpha, n}, L],\hspace{1cm}\frac{\partial R_{\beta}}{\partial t_{\alpha, n}}=[B_{\alpha, n}, R_{\beta}],\hspace{1cm}B_{\alpha, n}:=(L^n R_{\alpha})_+.
\end{equation*}
The  $``+"$ means the projection onto the terms with the nonnegative powers about $\d$.
One can check that the operator $\partial$ now is equal to $\partial_{t_{11}}+\ldots +\partial_{t_{N1}}$.
In fact the Lax operator $L$ and $R_{\alpha}$ can have the following dressing structures
\[L=S\d S^{-1}, \ \ R_{\alpha}=SE_{\alpha}S^{-1},\]
where
\[S=E+\sum_{i=1}^{\infty}S_i\d^{-i},\]
and $S_1$ will be used to define eigenfunction $P^{(r)}$ and adjoint eigenfunction $Q^{(r)}$ later in a constrained hierarchy.

As \cite{zhangJPA}, the following reduction condition is imposed onto the $N$-component KP hierarchy:
\[\label{constr}\sum_{l=1}^rd_l^s(LR_l)^s+\sum_{l=r+1}^NLR_l= \sum_{l=1}^rd_l^s(LR_l)_+^s+\sum_{l=r+1}^N(LR_l)_+,\ \ 1\leq r\leq N.\]

Let us denote following operators as \cite{zhangJPA}
\[\hat \d:=\d_1+\d_1+\dots+\d_r,\ \ \d_i=\frac{\d}{\d t_{i1}},\]
\[\hat L^{(i)}=S^{(r)}(\hat \d)E_{i}^{(r)}\hat \d (S^{(r)}(\hat \d))^{-1},\ \
\hat L=\sum_{i=1}^rd_i\hat L^{(i)},\]
and
\[B_k^{(l,r)}&=&(S^{(r)}(\hat \d)E_l^{(r)}\hat \d^k (S^{(r)}(\hat \d))^{-1})_+:=(C_k^{(l,r)})_+,\ \ 1\leq l\leq r,\\
C_k^{(l,r)}&=&S^{(r)}(\hat \d)E_l^{(r)}\hat \d^k (S^{(r)}(\hat \d))^{-1},\ \ 1\leq l\leq r,\]
\[P^{(r)}=(S_1)_{ij}|_{1\leq i\leq r,r+1\leq j\leq N},\ \ Q^{(r)}=(S_1)_{ij}|_{1\leq j\leq r,r+1\leq i\leq N},\]
where $S^{(r)}(E_{i}^{(r)})$ is the $r\times r$ principal sub-matrix of $S(E_{i})$ and  the projection $``+"$ here is about the operator $\hat \d.$
This further lead to
\[\L=\hat L^s=\sum_{l=1}^rd_l^sB_s^{(l,r)}+P^{(r)}\hat \d^{-1}Q^{(r)}.\]
Under the constraint eq.\eqref{constr}, the following evolutionary equations of the constrained $N$-component KP hierarchy can be derived \cite{zhangJPA}.

\begin{proposition}

For the constrained $N$-component KP hierarchy, the following hierarchies of equations can be derived \cite{zhangJPA}

\[\frac{\d P^{(r)}}{\d t_{nl}}=B_n^{(l,r)}P^{(r)},\ \ \frac{\d Q^{(r)}}{\d t_{nl}}=-(B_n^{(l,r)})^*Q^{(r)},\]
\[\frac{\d \hat L^s}{\d t_{nl}}=[B_n^{(l,r)},\hat L^s],\]
where for $B_n^{(l,r)}=\sum_{i=0}^nB_i\hat \d^i,$
\[(B_n^{(l,r)})^*Q^{(r)}:=\sum_{i=0}^n(-1)^i\hat \d^i(Q^{(r)}B_i).\]
\end{proposition}

When $r=1,N=2$, one can obtain the $s$-constrained KP hierarchy \cite{chengyiliyishenpla1991,chengyiliyishenjmp1995}. When  $r=1,N\geq 2$, one can obtain the vector $s$-constrained KP hierarchy\cite{vector}.
For general $r,N$, the generalized constrained KP hierarchy can be obtained which  can also be seen in the setting of the generalized Drinfeld-Sokolov construction \cite{generalized DS,generalizedDS}.

\section{Virasoro symmetry of the  constrained multi-component KP hierarchy}
In this section, we shall construct the additional symmetry and discuss the algebraic structure of the
additional symmetry  flows of the $N$-component  constrained KP hierarchy.

To this end, firstly we define $\hat\Gamma^{(r)}$ and the Orlov-Shulman's  operator $\M$ like in \cite{os}
 \begin{equation}
\hat\Gamma^{(r)}=\sum_{\beta=1}^rt_{1\beta}\frac{E_{\beta}}{sd_{\beta}^s}\hat \d^{1-s}+\sum_{\beta=1}^r\sum_{n=2}^{\infty}\frac nsd_{\beta}^{-s}E_{\beta}\hat \d^{n-s}t_{n\beta},\ \ \M= S^{(r)} \hat\Gamma^{(r)} (S^{(r)})^{-1}.
\end{equation}

It  is easy to find the following formula
\begin{equation}\label{commute}
[\partial_{t_{n\beta}}-\hat \d^nE_{\beta},\hat\Gamma^{(r)}]=0.
\end{equation}

There are the following commutation relations
\begin{equation}
[\sum_{\beta=1}^rd_{\beta}^s\hat \d^s E_{\beta},\hat\Gamma^{(r)}]=E,
\end{equation}
which can be verified by a straightforward calculation.
By using the Sato equation, the isospectral flow of the $\M$ operator
is given by
 \begin{equation}
\partial_{t_{n\beta}}\M=[B_n^{(\beta,r)},\M].
\end{equation}  More generally,
\begin{equation}
\partial_{t_{n\beta}}(\M^m\L^l)=[B_n^{(\beta,r)},\M^m\L^l].
\end{equation}

The Lax operator $\L$ and the Orlov-Shulman's $\M$ operator satisfy the following canonical relation
\[[\L,\M]=E.\] Then the additional flows for the time variable $t_{m,n,\beta}$ will be
defined as follows
\begin{equation}
\dfrac{\partial S}{\partial t_{m,n,\beta}}=-(\M^mC_n^{(\beta,r)})_-S,\ \ m,n \in \N, 1\leq \beta\leq r.
\end{equation}
or equivalently
\begin{equation}
\dfrac{\partial \L}{\partial t_{m,n,\beta}}=-[(\M^mC_n^{(\beta,r)})_-,\L], \qquad
\dfrac{\partial \M}{\partial t_{m,n,\beta}}=-[(\M^mC_n^{(\beta,r)})_-,\M].
\end{equation}
Later we can prove the additional flows $\dfrac{\partial }{\partial
t_{m,n,\beta}}$  commute with the flow $\dfrac{\partial
}{\partial t_{k,\gamma}}$, i.e. $[\dfrac{\partial }{\partial
t_{m,n,\beta}},\dfrac{\partial
}{\partial t_{k,\gamma}}]=0$, but do not
commute with each other. They form a kind of $W_{\infty}$ infinite dimensional  additional Lie algebra symmetries which contain Virasoro algebra as a special
subalgebra. To this purpose, we need several lemmas and propositions as preparation firstly.

For above local differential operators $B_n^{(\beta,r)}$, we have the following lemma.
\begin{lemma}
  \label{lemm} $[B_n^{(\beta,r)},P^{(r)}\hat \d^{-1}Q^{(r)}]_-=B_n^{(\beta,r)}(P^{(r)})\hat \d^{-1}Q^{(r)}-P^{(r)}\hat \d^{-1}B_n^{(\beta,r)*}(Q^{(r)})$.
\end{lemma}
\begin{proof}
 Firstly we consider a fundamental monomial: $\hat \d^n$ ($n\ge 1$). Then
  \begin{displaymath}
    [\hat \d^n,P^{(r)}\hat \d^{-1}Q^{(r)}]_-=(\hat \d^n(P^{(r)}))\hat \d^{-1}Q^{(r)}- (P^{(r)}\hat \d^{-1}Q^{(r)} \hat \d^n)_-.
  \end{displaymath}
  Notice that the second term can be rewritten in the following way
  \begin{eqnarray*}
    (P^{(r)}\hat \d^{-1}Q^{(r)} \hat \d^{n})_-
    =(P^{(r)}Q^{(r)} \hat \d^{n-1}
        -P^{(r)}\hat \d^{-1}(\hat \d Q^{(r)})\hat \d^{n-1})_- \\
    =(P^{(r)}\hat \d^{-1}(-\hat \d Q^{(r)})\hat \d^{n-1})_-
    =\cdots
    =P^{(r)}\hat \d^{-1}\left((-\hat \d)^n
    (Q^{(r)})\right)=P^{(r)}\hat \d^{-1}(\hat \d^n)^*(Q^{(r)}),
  \end{eqnarray*}
  then the lemma is proved.
\end{proof}

We can get some properties of the Lax operator in the following proposition.

\begin{lemma}The Lax operator $\L$ of the constrained $N$-component KP hierarchy  will satisfy the relation of
\begin{equation}\label{Lk}
(\L^k)_-=\sum_{j=0}^{k-1}\L^{k-j-1}(P^{(r)})\hat \d^{-1}{(\L^*)}^j(Q^{(r)}),\ \ k\in \Z.
\end{equation}
where $\L(\phi):=(\L)_{+}(\phi) + P^{(r)}\hat \d^{-1}(Q^{(r)}\phi),$  for arbitrary $r\times r$ matrix function $\phi$.
\end{lemma}

The action of original
additional flows of the  constrained $N$-component KP hierarchy  is  expressed by
\begin{equation}\label{additional}
(\partial_{1,k,\beta} \L)_-=[(\M C_k^{(\beta,r)})_+,\L]_-+(C_k^{(\beta,r)})_-.
\end{equation}

To keep the consistency with flow equations on eigenfunction and adjoint eigenfunction $P^{(r)},Q^{(r)}$, we shall introduce an operator
$T_k^{(\alpha,r)}$ as following to modify the additional symmetry
of the constrained multi-component KP hierarchy.

We now introduce a pseudo  differential operator  $T_k^{(\alpha,r)}$,
\begin{eqnarray}
T_k^{(\alpha,r)}&=&0,k=-1,0,1,\label{y12}\\
T_k^{(\alpha,r)}&=&\sum_{j=0}^{k-1}[j-\frac{1}{2}(k-1)]C_{k-1-j}^{(\alpha,r)}(P^{(r)})
\hat \d^{-1}(C_{j}^{(\alpha,r)*})(Q^{(r)}),k\geq 2.\label{y}
\end{eqnarray}

The following
lemmas are
necessary to concern the Virasoro symmetry.
\begin{lemma}If $
X=M\hat \d^{-1}N,$ then
\begin{equation}\label{operatorXL}
[X,\L]_-=[M\hat \d^{-1}\L^*(N)-\L(M)\hat \d^{-1}N]+ [X(P^{(r)})\hat \d^{-1}Q^{(r)}-P^{(r)}\hat \d^{-1}X^*(Q^{(r)})].
\end{equation}
\end{lemma}
\begin{lemma}
The action of flows $\partial_{t_{l,\beta}}$ of the  constrained $N$-component KP hierarchy  on the
$T_k^{(\alpha,r)}$ is
\begin{equation}\label{ykderivat}
\partial_{t_{l,\beta}}
T_k^{(\alpha,r)}=[B_l^{(\beta,r)},T_k^{(\alpha,r)}]_-.
\end{equation}
\end{lemma}
\begin{proof}
\begin{eqnarray*}
\partial_{t_{l,\beta}}T_k^{(\alpha,r)}&=&\partial_{t_{l,\beta}}(\sum_{j=0}^{k-1}[j-\frac{1}{2}(k-1)]C_{k-1-j}^{(\alpha,r)}(P^{(r)})\hat \d^{-1}(C_{j}^{(\alpha,r)*})(Q^{(r)}))\\\nonumber
&=&\sum_{j=0}^{k-1}[j-\frac{1}{2}(k-1)]\{\partial_{t_{l,\beta}}(C_{k-1-j}^{(\alpha,r)}(P^{(r)}))\hat \d^{-1}(C_{j}^{(\alpha,r)*})(Q^{(r)})\\
&&+C_{k-1-j}^{(\alpha,r)}(P^{(r)})\hat \d^{-1}\partial_{t_{l,\beta}}((C_{j}^{(\alpha,r)*})(Q^{(r)}))\}\\\nonumber
&=&[B_l^{(\beta,r)}\circ \sum_{j=0}^{k-1}[j-\frac{1}{2}(k-1)]C_{k-1-j}^{(\alpha,r)}(P^{(r)})\hat \d^{-1}(C_{j}^{(\alpha,r)*})(Q^{(r)})]_-\\
&&-[(\sum_{j=0}^{k-1}[j-\frac{1}{2}(k-1)]C_{k-1-j}^{(\alpha,r)}(P^{(r)})\hat \d^{-1}(C_{j}^{(\alpha,r)*})(Q^{(r)}))\circ B_l^{(\beta,r)}]_-\\
&=&[B_l^{(\beta,r)},(\sum_{j=0}^{k-1}[j-\frac{1}{2}(k-1)]C_{k-1-j}^{(\alpha,r)}(P^{(r)})\hat \d^{-1}(C_{j}^{(\alpha,r)*})(Q^{(r)}))]_-\\
&=&[B_l^{(\beta,r)},T_k^{(\alpha,r)}]_-.
\end{eqnarray*}
\end{proof}

Further, the following expression of $[T_{k-1}^{(\beta,r)},\L]_-$ is also
necessary to define the additional flows of the  constrained $N$-component KP hierarchy.
\begin{lemma}
The Lax operator $\L$ of the constrained $N$-component KP hierarchy and $T_{k-1}^{(\beta,r)}$ has the following relation,
\begin{eqnarray}\notag
[T_{k-1}^{(\beta,r)},\L]_-&=&-(C_{k}^{(\beta,r)})_-+\frac{k}{2}[P^{(r)}\hat \d^{-1}(C_{k-1}^{(\beta,r)})^*(Q^{(r)})+C_{k-1}^{(\beta,r)}(P^{(r)})\hat \d^{-1}Q^{(r)}\\
\label{ykl}
&&+T_{k-1}^{(\beta,r)}(P^{(r)})\hat \d^{-1}Q^{(r)}-P^{(r)}\hat \d^{-1}(T_{k-1}^{(\beta,r)})^*(Q^{(r)})].
\end{eqnarray}
\end{lemma}
\begin{proof}
A direct calculation can lead to
\begin{eqnarray*}
[T_{k-1}^{(\beta,r)},\L]_-&=&[\sum^{k-2}_{j=0}[j-\frac{1}{2}(k-2)]C_{k-2-j}^{(\beta,r)}(P^{(r)})\hat \d^{-1}(C_{j}^{(\beta,r)})^*(Q^{(r)}),\L]_-\\\nonumber
&=&\sum^{k-2}_{j=0}[j-\frac{1}{2}(k-2)]C_{k-2-j}^{(\beta,r)}(P^{(r)})\hat \d^{-1}(C_{j+1}^{(\beta,r)})^*(Q^{(r)})\\\nonumber
&&-\sum^{k-2}_{j=0}[j-\frac{1}{2}(k-2)]C_{k-1-j}^{(\alpha,r)}(P^{(r)})\hat \d^{-1}(C_{j}^{(\beta,r)})^*(Q^{(r)})\\\nonumber
&&+(T_{k-1}^{(\beta,r)}(P^{(r)})\hat \d^{-1}Q^{(r)}-P^{(r)}\hat \d^{-1}(T_{k-1}^{(\beta,r)})^*(Q^{(r)}))\\\nonumber
&=&-\sum^{k-2}_{j=1}C_{k-1-j}^{(\beta,r)}(P^{(r)})\hat \d^{-1}(C_{j}^{(\beta,r)})^*(Q^{(r)})\\\nonumber
&&+(\frac{k}{2}-1)[P^{(r)}\hat \d^{-1}{(C_{k-1}^{(\beta,r)}}^*(Q^{(r)})+C_{k-1}^{(\beta,r)}(P^{(r)})\hat \d^{-1}Q^{(r)}]\\\nonumber
&&+(T_{k-1}^{(\beta,r)}(P^{(r)})\hat \d^{-1}Q^{(r)}-P^{(r)}\hat \d^{-1}(T_{k-1}^{(\beta,r)})^*(Q^{(r)})),
\end{eqnarray*}
which further help us deriving eq.\eqref{ykl} using eq.\eqref{Lk}.
\end{proof}

Putting together (\ref{additional}) and (\ref{ykl}), we define the
additional flows of the  constrained $N$-component KP hierarchy  as
 \begin{equation}\label{tkflow}
\partial_{t_{1,k,\beta}}\L=[-(\M C_k^{(\beta,r)})_-+T_{k-1}^{(\beta,r)},\L],
\end{equation}
where
$T_{k-1}^{(\beta,r)}=0$, for $k=0,1,2$, such that the right-hand side of
(\ref{tkflow}) is in the form of derivation of Lax equations. Generally, one can also derive
\begin{equation}\label{MLK}
\partial_{t_{1,k,\beta}}(\M\L^l)=[-(\M C_k^{(\beta,r)})_-+T_{k-1}^{(\beta,r)},\M\L^l].
\end{equation}

Now we calculate the action of the additional flows eq.\eqref{tkflow}
on the eigenfunction $P^{(r)}$ and $Q^{(r)}$ of the  constrained $N$-component KP hierarchy.
\begin{theorem}\label{symmetre}
The acting of additional flows of  constrained $N$-component KP hierarchy on the eigenfunction $P^{(r)}$ and $Q^{(r)}$ are
\begin{equation}\label{BAfunction}
\begin{split}
{\partial_{t_{1,k,\beta}}P^{(r)}}&=(\M
 C_k^{(\beta,r)})_+(P^{(r)})+T_{k-1}^{(\beta,r)}(P^{(r)})+\frac{k}{2} C_{k-1}^{(\beta,r)}(P^{(r)}),\\
 {\partial_{t_{1,k,\beta}}Q^{(r)}}&=-(\M
 C_k^{(\beta,r)})^*_+(Q^{(r)})-(T_{k-1}^{(\beta,r)})^*(Q^{(r)})+\frac{k}{2}{ C_{k-1}^{(\beta,r)*}}(Q^{(r)}).
\end{split}
\end{equation}
\end{theorem}
\begin{proof}
Substitution of (\ref{ykl}) to negative part of (\ref{tkflow}) shows
\begin{eqnarray}\label{tkflow2}
\begin{split}
{(\partial_{t_{1,k,\beta}}\L)}_-&=(\M  C_k^{(\beta,r)})_+(P^{(r)})\hat \d^{-1}(Q^{(r)})-P^{(r)}\hat \d^{-1}(\M  C_k^{(\beta,r)})_+^*(Q^{(r)})+T_{k-1}^{(\beta,r)}(P^{(r)})\hat \d^{-1}Q^{(r)}\\
&-P^{(r)}\hat \d^{-1}(T_{k-1}^{(\beta,r)})^*(Q^{(r)})
+\frac{k}{2}P^{(r)}\hat \d^{-1}(T_{k-1}^{(\beta,r)})^*(Q^{(r)})+\frac{k}{2}C_k^{(\beta,r)}(P^{(r)})\hat \d^{-1}Q^{(r)}.
\end{split}
\end{eqnarray}

On the other side,
\begin{equation}\label{tkl}
{(\partial_{t_{1,k,\beta}}\L)}_-=(\partial_{t_{1,k,\beta}}P^{(r)})\hat \d^{-1}Q^{(r)}+P^{(r)}\hat \d^{-1}{(\partial_{t_{1,k,\beta}}Q^{(r)})}.
\end{equation}
Comparing right hand sides of (\ref{tkflow2}) and (\ref{tkl})
implies  the action of additional flows on the eigenfunction and the
adjoint eigenfunction (\ref{BAfunction}).
\end{proof}

Next we shall prove the commutation relation between the additional
flows $\partial_{t_{1,k,\beta}}$ of  constrained $N$-component KP hierarchy and the original flows $\partial_{t_{l,\alpha}}$ of
it.
\begin{theorem}
The additional flows of $\partial_{t_{1,k,\beta}}$ are  symmetry flows of the  constrained $N$-component KP hierarchy, i.e. they commute with all $\partial_{t_{l,\alpha}}$ flows of the   constrained $N$-component KP hierarchy.
\end{theorem}
\begin{proof}
According the action of  $\partial_{t_{1,k,\beta}}$ and $\partial_{t_{l,\alpha}}$ on the
dressing operator $S$, then
\begin{eqnarray*}
[\partial_{t_{1,k,\beta}},\partial_{t_{l,\alpha}}]S &=& -\partial_{t_{1,k,\beta}}((C_{l}^{(\alpha,r)})_-
S)-\partial_{t_{l,\alpha}}[-(\M C_{k}^{(\beta,r)})_-+T_{k-1}^{(\beta,r)}]S \\
&=&(-\partial_{t_{1,k,\beta}}C_{l}^{(\alpha,r)})_-S-(C_{l}^{(\beta,r)})_-\partial_{t_{1,k,\beta}}S-[(\M
C_{k}^{(\beta,r)})_--T_{k-1}^{(\beta,r)}](C_{l}^{(\alpha,r)})_-S \\
 &&+[(C_{l}^{(\alpha,r)})_+,\M C_{k}^{(\beta,r)}]_-S-(\partial_{t_{l,\alpha}} T_{k-1}^{(\beta,r)})S\\
&=&[(C_{l}^{(\alpha,r)})_-,-T_{k-1}^{(\beta,r)}]_-S+[-T_{k-1}^{(\beta,r)},C_{l}^{(\alpha,r)}]_-S-(\partial_{t_{l,\alpha}} T_{k-1}^{(\beta,r)})S \\
&=&[B_{l}^{(\alpha,r)},T_{k-1}^{(\beta,r)}]_-S-(\partial_{t_{1,\alpha}} T_{k-1}^{(\beta,r)})S\\
&=&0.
\end{eqnarray*}
Therefore the theorem holds.
\end{proof}

Taking into account $T_{k-1}^{(\beta,r)}=0$  for $k=0,1,2 $, then
eq.\eqref{BAfunction} becomes
\begin{eqnarray}\label{PLqr}
\begin{split}
\partial_{t_{1,l,\beta}}P^{(r)}&=(\M
C_{l}^{(\beta,r)})_+(P^{(r)})+\frac{1}{2}l C_{l-1}^{(\beta,r)} P^{(r)},\ \ l=0,1,2,\\
\partial_{t_{1,l,\beta}}Q^{(r)}&=-(\M
C_{l}^{(\beta,r)})^*_+(Q^{(r)})+\frac{1}{2}l {C_{l-1}^{(\beta,r)}}^* Q^{(r)}, \ \ l=0,1,2.
\end{split}
\end{eqnarray}

Then  using eq.\eqref{PLqr} and the
relation
${{\partial_{t_{1,l,\beta}}}(\L^k(P^{(r)}))}=({\partial_{t_{1,l,\beta}}}(\L^k))(P^{(r)})+\L^k
{\partial_{t_{1,l,\beta}}}(P^{(r)})$,
we can find
the additional  flows $\partial_{t_{1,l,\beta}}$ of  constrained $N$-component KP hierarchy have the following relations
\begin{eqnarray}\label{lqstar}
\begin{split}
{\partial_{t_{1,l,\beta}} C_{k}^{(\alpha,r)}(P^{(r)})}&=(\M
C_{l}^{(\beta,r)})_+(C_{k}^{(\alpha,r)}(P^{(r)}))+(k+\frac{l}{2})C_{k+l-1}^{(\alpha,r)}\delta_{\alpha,\beta}(P^{(r)})+T_{l-1}^{(\beta,r)}C_{k}^{(\alpha,r)}(P^{(r)}),\\
{\partial_{t_{1,l,\beta}} {C_{k}^{(\alpha,r)}}^*(Q^{(r)})}&=-(\M
C_{l}^{(\beta,r)})^*_+{C_{k}^{(\alpha,r)}}^*(Q^{(r)})+(k+\frac{l}{2}){C_{k+l-1}^{(\alpha,r)}}^*\delta_{\alpha,\beta}(Q^{(r)})
+(T_{l-1}^{(\beta,r)}C_{k}^{(\alpha,r)})^*(P^{(r)}).
\end{split}
\end{eqnarray}

Moreover, the action of
 $\partial_{t_{1,l,\beta}}$ on $T_k^{(\alpha,r)}$ is given by the following lemma.
\begin{lemma}
The actions on $T_k^{(\alpha,r)}$ of the additional  symmetry flows
$\partial_{t_{1,l,\beta}}$ of the  constrained $N$-component KP hierarchy are
\begin{eqnarray}\label{PL}
\partial_{t_{1,l,\beta}} T_k^{(\alpha,r)}=[(\M C_l^{\beta,r})_++T_{l-1}^{(\beta,r)},T_k^{(\alpha,r)}]_-+(k-l+1)T_{k+l-1}^{(\alpha,r)}\delta_{\alpha,\beta}.
\end{eqnarray}
\end{lemma}
\begin{proof} Using eq.\eqref{lqstar}, a straightforward calculation implies
\begin{eqnarray}\label{partialyk}
\begin{split}
\partial_{t_{1,l,\beta}} T_k^{(\alpha,r)}&=\partial_{t_{1,l,\beta}} (\sum_{j=0}^{k-1}[j-\frac{1}{2}(k-1)]C_{k-1-j}^{(\alpha,r)}(P^{(r)})\hat \d^{-1}(C_{j}^{(\alpha,r)*})(Q^{(r)}))\\
&=\sum_{j=0}^{k-1}[j-\frac{1}{2}(k-1)](\partial_{t_{1,l,\beta}}(C_{k-1-j}^{(\alpha,r)}(P^{(r)}))\hat \d^{-1}(C_{j}^{(\alpha,r)*})(Q^{(r)})\\
&+C_{k-1-j}^{(\alpha,r)}(P^{(r)})\hat \d^{-1}(\partial_{t_{1,l,\beta}}(C_{j}^{(\alpha,r)*})(Q^{(r)})))\\
&=\sum_{j=0}^{k-1}[j-\frac{1}{2}(k-1)][(\M C_l^{\beta,r})_++T_{l-1}^{(\beta,r)}](C_{k-1-j}^{(\alpha,r)}(P^{(r)}))\hat \d^{-1}(C_l^{\beta,r})^*(Q^{(r)})\\
&+\sum_{j=0}^{k-1}[j-\frac{1}{2}(k-1)](k-j-1+\frac{l}{2})C_{k+l-2-j}^{\alpha,r}(P^{(r)})\hat \d^{-1}(C_j^{\beta,r})^*(Q^{(r)})\delta_{\alpha,\beta}\\
&-\sum_{j=0}^{k-1}[j-\frac{1}{2}(k-1)]C_{k-1-j}^{(\alpha,r)}(P^{(r)})\hat \d^{-1}[(\M C_l^{\beta,r})^*_++(T_{l-1}^{(\beta,r)})^*](C_j^{\beta,r})^*(Q^{(r)})\\
&+\sum_{j=0}^{k-1}[j-\frac{1}{2}(k-1)](j+\frac{l}{2})C_{k-1-j}^{(\alpha,r)}(P^{(r)})\hat \d^{-1}(C_{j+l-1}^{\beta,r})^*(Q^{(r)})\delta_{\alpha,\beta},
\end{split}
\end{eqnarray}
which can be simplified  to eq.\eqref{PL}.
\end{proof}

Now it is time to identity the algebraic structure of the
additional symmetry flows of the  constrained $N$-component KP hierarchy.
\begin{theorem}\label{alg}
The additional flows $\partial_{t_{1,k,\beta}}$ of the  constrained $N$-component KP hierarchy form the
positive half of Virasoro algebra, i.e., for $l=0,1,2; k\geq 0,1\leq \alpha,\beta\leq r,$
\begin{equation}
[\partial_{t_{1,l,\alpha}},\partial_{t_{1,k,\beta}}]=(k-l)\delta_{\alpha,\beta}\partial_{t_{1,k+l-1,\alpha}}.
\end{equation}
\end{theorem}
\begin{proof}
Using $[(\M C_{l}^{(\alpha,r)})_-,T_{k-1}^{(\beta,r)}]_-=[(\M C_{l}^{(\alpha,r)})_-,T_{k-1}^{(\beta,r)}]$
and the Jacobi identity, one can derive the following computation

\begin{eqnarray*}
&&[\partial_{t_{1,l,\alpha}},\partial_{t_{1,k,\beta}}]\L\\
&=&\partial_{t_{1,l,\alpha}}([-(\M C_{k}^{(\beta,r)})_-,\L]+[T_{k-1}^{\beta,r},\L])-\partial_{t_{1,k,\beta}}([-(\M
 C_l^{\alpha,r})_-,\L]+[T_{l-1}^{\alpha,r},\L])) \\
&=&\partial_{t_{1,l,\alpha}} [-(\M C_{k}^{(\beta,r)})_-,\L] +[\partial_{t_{1,l,\alpha}} T_{k-1}^{\beta,r},\L]+[T_{k-1}^{\beta,r},\partial_{t_{1,l,\alpha}} \L]+[\partial_{t_{1,k,\beta}}(\M
 C_l^{\alpha,r})_-,\L]\\
&&+[(\M
 C_l^{\alpha,r})_-,\partial_{t_{1,k,\beta}} \L]-[\partial_{t_{1,k,\beta}} T_{l-1}^{\alpha,r},\L]-[T_{l-1}^{\alpha,r},\partial_{t_{1,k,\beta}} \L]\\
&=&[[(\M C_{l}^{(\alpha,r)})_-,\M C_k^{\beta,r}]_-,\L] +[-(\M C_{k}^{(\beta,r)})_-,[-(\M C_l^{(\alpha,r)})_-
,\L]] \\
&&+[[(\M C_l^{\alpha,r})_+,T_{k-1}^{\beta,r}]_-+(k-l)T_{k+l-2}^{(\beta,r)}\delta_{\alpha,\beta},\L]\\
&&+[T_{k-1}^{\beta,r},[-(\M C_l^{\alpha,r})_-,\L]]
+[[-(\M C_{k}^{(\beta,r)})_-+T_{k-1}^{\beta,r},\M C_l^{\alpha,r}]_-,\L]\\
&&+[(\M
C_l^{\alpha,r})_-,[-(\M C_{k}^{(\beta,r)})_-+T_{k-1}^{\beta,r},\L]]\\
&=&[[(\M C_{l}^{(\alpha,r)})_-,\M C_k^{\beta,r}]_-,\L] +[-(\M C_{k}^{(\beta,r)})_-,[-(\M C_l^{(\alpha,r)})_-
,\L]] \\
&&+[(k-l)T_{k+l-2}^{(\beta,r)}\delta_{\alpha,\beta},\L]+[[-(\M C_{k}^{(\beta,r)})_-,\M C_l^{\alpha,r}]_-,\L]\\
&&+[(\M
C_l^{\alpha,r})_-,[-(\M C_{k}^{(\beta,r)})_-,\L]]\\
&=&(k-l)[-\delta_{\alpha,\beta}(\M C_{l+k-1}^{(\alpha,r)})_-,\L] +[(k-l)T_{k+l-2}^{(\beta,r)}\delta_{\alpha,\beta},\L]\\
&=&(k-l)\delta_{\alpha,\beta}\partial_{t_{1,k+l-1,\alpha}}\L.
\end{eqnarray*}
\end{proof}

We define
$L_{k-1,\alpha}:=\partial_{t_{1,k,\alpha}},$ and obtain recursively all higher Virasoro operators as
\[L_{k+1,\alpha}\delta_{\alpha,\beta}=-\frac{1}{k-1}[L_{k,\alpha},L_{1,\beta}],\ \ k\geq 2.\] By induction as
\[[L_{m,\alpha},L_{k+1,\beta}]=\frac{m-1}{k-1}[L_{k,\beta},L_{m+1,\alpha}]+\frac{k-m}{k-1}[L_{1,\alpha},L_{m+k,\beta}],\]
the obtained higher Virasoro operators in the above recursive way form a closed matrix-formed Virasoro algebra up to irrelevant terms containing flows which commute with the $sl(2)$ additional symmetry generators $\{L_{-1,\alpha},L_{0,\alpha},L_{1,\alpha};\ \ 1\leq \alpha\leq r\}$.

\section{ The discrete constrained multi-component KP hierarchy}
At the beginning  of the this section, to make this paper
be self-contained, let us firstly recall   some useful
facts of discrete calculus.

For any function $f(n)$, the discrete differential $\Delta$ is defined by
\begin{equation}
   \Delta(f(n))=f(n+1)-f(n), \label{q-derivative}
   \end{equation}
and the discrete shift operator $\Lambda$ is defined by
$
  \Lambda(f(n))=f(n+1).$
 We denote the formal inverse of
$\Delta$ as $\Delta^{-1}$. The following discrete deformed
Leibnitz rule holds
\begin{equation}
     \Delta^n \circ f=\sum_{k\ge0}\binom{n}{k}\Lambda^{n-k}(\Delta^kf)\Delta^{n-k},\qquad n\in \Z.
     \end{equation}

 For any two discrete pseudo-differential operators $W=\sum\limits_ip_i(n)\Delta^i$, the conjugate operation ``$*$''
for $W$ is defined by $W^*=\sum\limits_i(\Delta^*)^ip_i(n)$ with
\[&&\Delta^*=-\Delta\Lambda^{-1},\ \ (\Delta^{-1})^*=(\Delta^*)^{-1}=-\Lambda
\Delta^{-1}.\]

Now we prove some useful properties for the operators which are used later. \\
{\sl {\bf  Lemma 2.1}}  For arbitrary function $f$ and $\Delta$, $\Lambda$ as
above, the following identities hold.
\begin{align}
&(1)\quad \Delta\circ\Lambda=\Lambda\circ\Delta,  \\
&(2)\quad \Delta^{-1}\circ
f\circ\Delta^{-1}=(\Delta^{-1}f)\circ\Delta^{-1}-\Delta^{-1}\circ
\Lambda(\Delta^{-1} f),\label{85}\\
&(3)\quad \Delta^{-1}\circ
f=\Lambda^{-1}(f)\circ\Delta^{-1}+\Delta^{-1}\circ
(\Delta^* f)\circ\Delta^{-1}.\label{86}
\end{align}
\begin{proof}
The proof is standard and direct. We omit it here.
\end{proof}
Because of the above lemma, for arbitrary function $f_1,g_1,f_2,g_2$, the following formula can be derived
\begin{equation}
\label{Drelation}f_1\Delta^{-1}g_1\circ f_2\Delta^{-1}g_2=f_1\Delta^{-1}(g_1,f_2)\Delta^{-1}g_2-f_1\Delta^{-1}\circ \Lambda(\Delta^{-1}(g_1,f_2))g_2.\end{equation}

\begin{proposition}\label{Dformula}
The following identities hold:
\[
res((\Delta A)(n))&=&res(\Delta \circ A(n)-A(n+1)\circ \Delta),\\
\label{Dleq0}
P_{< 0}(\Delta^{-1}A)
&=& \Delta^{-1} P_{< 0}(A)+ \Delta^{-1}P_{ 0}(A^*),\\
\label{Dresp0}
res(\Delta^{-1} A)&=&(\Lambda^{-1}P_0(A^*)),
\]
where $(\Delta A)$ denotes the action of $\Delta $ on operator $A$, $(\Lambda^{-1}P_0(A^*))$ means a backward shift of function $P_0(A^*)$(zero order term of operator $A^*$ on $\Delta$) on $x$.
\end{proposition}
\begin{proof}
One can check the proof in \cite{ghostdKP}.
\end{proof}

After integrable discretization by $u(n)=u(x+n\epsilon)$, one can get similar results as constrained multi-component KP hierarchy for its discrete version.

Using the above proposition, the following proposition can be easily got similarly as \cite{ghostdKP}.
\begin{proposition}
If $\alpha$ and $\beta$ are two local discrete  operators, then
\begin{eqnarray*}res(\Delta^{-1}\alpha \beta \Delta^{-1} )&=&res(\Delta^{-1}P_0(\alpha^*) \beta \Delta^{-1} )+res(\Delta^{-1}\alpha P_0(\beta) \Delta^{-1} ).
\end{eqnarray*}
\end{proposition}
\begin{proof}
See the proof in \cite{ghostdKP}.
\end{proof}

The Lax operator $L_{\Delta}$ of discrete multicomponent KP hierarchy is given by
\begin{equation}\label{Dqkplaxoperator}
L_{\Delta}=E\Delta+ v_0 +
v_{-1}\Delta^{-1}+v_{-2}\Delta^{-2}+\cdots.
\end{equation}
where $v_i=v_i(n,t_{\alpha, n},\cdots,),i=0,-1,-2, -3, \cdots $.
 are $N\times N$ matrices.

  There are another $N$ pseudo difference operators $R_{\Delta1},\ldots, R_{\Delta N}$ of the form
$$R_{\Delta\alpha}=E_{\alpha}+v_{\alpha 1}\Delta^{-1}+v_{\alpha 2}\Delta^{-2}+\ldots,$$ where $E_{\alpha}$ is the $N\times N$ matrix with $1$ on the $(\alpha,\alpha)$-component and zero elsewhere, and $v_{\alpha 1}, v_{\alpha 2},\ldots$ are also $N\times N$ matrices. The operators $L_{\Delta}, R_{\Delta 1},\ldots, R_{\Delta N}$ satisfy the following conditions:
$$L_{\Delta}R_{\Delta\alpha}=R_{\Delta\alpha} L_{\Delta}, \quad R_{\Delta\alpha}R_{\Delta\beta}=\delta_{\alpha\beta}R_{\Delta\alpha}, \quad \sum_{\alpha=1}^N R_{\Delta\alpha}= E.$$

The Lax equations of the discrete multi-component KP (dmKP) hierarchy are:
\begin{equation*}
\frac{\d L_{\Delta}}{\d t_{\alpha, n}}=[B_{\Delta \alpha, n}, L_{\Delta}],\hspace{1cm}\frac{\d R_{\Delta\beta}}{\d t_{\alpha, n}}=[B_{\Delta \alpha, n}, R_{\Delta\beta}],\hspace{1cm}B_{\Delta \alpha, n}:=(L_{\Delta}^n R_{\Delta\alpha})_+.
\end{equation*}

In fact the Lax operator $L_{\Delta}$ in
eq.(\ref{Dqkplaxoperator}) and $R_{\Delta\alpha}$ can have the following dressing structures
\[L_{\Delta}=S_{\Delta}\Delta S_{\Delta}^{-1}, \ \ R_{\Delta\alpha}=S_{\Delta}E_{\alpha}S_{\Delta}^{-1},\]
where
\[S_{\Delta}=E+\sum_{i=1}^{\infty}S_{\Delta i}\Delta^{-i}.\]
Dressing operator $S_{\Delta}$ satisfies Sato equation
\begin{equation}
\dfrac{\partial S_{\Delta}}{\d t_{\alpha, n}}=-(L_{\Delta}^n R_{\Delta\alpha})_-S_{\Delta}, \quad 1\leq\alpha\leq N,n=1,2, 3, \cdots.
\end{equation}

Considering one kind of reduction of the discrete multi-component KP hierarchy, the hierarchy will become the  discrete constrained multi-component KP hierarchy in the next subsection.

As \cite{zhangJPA}, the following reduction condition is imposed onto the discrete $N$-component KP hierarchy:
\[\label{Dconstr}\sum_{l=1}^rd_l^s(L_{\Delta}R_{\Delta l})^s+\sum_{l=r+1}^NL_{\Delta}R_{\Delta l}= \sum_{l=1}^rd_l^s(L_{\Delta}R_{\Delta l})_+^s+\sum_{l=r+1}^N(L_{\Delta}R_{\Delta l})_+,\ \ 1\leq r\leq N.\]

Let us denote following operators as \cite{zhangJPA}
\[\hat \Delta:=\d_1+\d_1+\dots+\d_r,\ \ \d_i=\frac{\d}{\d t_{i,1}},\]
\[\hat L_{\Delta}^{(i)}=S_{\Delta}^{(r)}(\hat \Delta)E_{i}^{(r)}\hat \Delta (S_{\Delta}^{(r)}(\hat \Delta))^{-1},\ \
\hat L_{\Delta}=\sum_{i=1}^rd_i\hat L_{\Delta}^{(i)},\]
and
\[B_{\Delta k}^{(l,r)}=(S_{\Delta}^{(r)}(\hat \Delta)E_l^{(r)}\hat \Delta^k (S_{\Delta}^{(r)}(\hat \Delta))^{-1})_+=(C_{\Delta k}^{(l,r)})_+,\ \ 1\leq l\leq r,\]
\[\phi^{(r)}_{\Delta}=(S_{\Delta 1})_{ij}|_{1\leq i\leq r,r+1\leq j\leq N},\ \ \psi^{(r)}_{\Delta}=(S_{\Delta 1})_{ij}|_{1\leq j\leq r,r+1\leq i\leq N},\]
where $S_{\Delta}^{(r)}$ is the $r\times r$ principal sub-matrix of $S_{\Delta}$ and  the projections here is about the operator $\hat \Delta.$
This further lead to
\[\L_{\Delta}=\hat L_{\Delta}^s=\sum_{l=1}^rd_l^sB_{\Delta s}^{(l,r)}+\phi^{(r)}\hat \Delta^{-1}\psi^{(r)}.\]
Under the constraint eq.\eqref{Dconstr}, the following evolutionary equations of the discrete constrained $N$-component KP hierarchy can be derived.

\begin{proposition}

For discrete constrained $N$-component KP hierarchy, the following hierarchies of equations can be derived \cite{zhangJPA}

\[\frac{\d \phi^{(r)}}{\d t_{n,l}}=B_{\Delta n}^{(l,r)}\phi^{(r)},\ \ \frac{\d \psi^{(r)}}{\d t_{n,l}}=-(B_{\Delta n}^{(l,r)})^*\psi^{(r)},\]
\[\frac{\d \L_{\Delta}}{\d t_{n,l}}=[B_{\Delta n}^{(l,r)},\L_{\Delta}],\]
where for $B_{\Delta n}^{(l,r)}=\sum_{i=0}^nB_i\hat \Delta^i,$
\[(B_{\Delta n}^{(l,r)})^*\psi^{(r)}:=\sum_{i=0}^n(-1)^i\hat \Delta^i\Lambda^{-i}(\psi^{(r)}B_i).\]
\end{proposition}

When $r=1,N=2$, one can obtain the discrete $s$-constrained KP hierarchy. When  $r=1,N\geq 2$, one can call it the discrete vector $s$-constrained KP hierarchy.
For general $r,N$, it is about the generalized discrete constrained KP hierarchy.

\section{Virasoro symmetry of the  discrete constrained multi-component KP hierarchy}
In this section, we shall construct the additional symmetry and discuss the algebraic structure of the
additional symmetry  flows of the multi-component discrete constrained KP hierarchy.

To this end, firstly we define $\hat \Gamma^{(r)}_{\Delta}$ and the Orlov-Shulman's  operator $\M_{\Delta}$ as

 \begin{equation}
\hat\Gamma^{(r)}_{\Delta}=\sum_{\beta=1}^rt_{1\beta}\frac{E_{\beta}}{sd_{\beta}^s}\hat \Delta^{1-s}+\sum_{\beta=1}^r\sum_{n=2}^{\infty}\frac nsd_{\beta}^{-s}E_{\beta}\hat \Delta^{n-s}t_{n\beta},\ \ \M_{\Delta}= S^{(r)} \hat\Gamma^{(r)}_{\Delta} (S^{(r)})^{-1}.
\end{equation}

It  is easy to find the following formula
\begin{equation}\label{Dcommute}
[\partial_{t_{k,\beta}}-\hat \Delta^kE_{\beta},\hat\Gamma^{(r)}_{\Delta}]=0.
\end{equation}
Define another operator
\begin{equation}\label{D}
\M_{\Delta}=W_{\Delta} \circ \hat\Gamma^{(r)}_{\Delta}\circ W_{\Delta}^{-1}.
\end{equation}
There are the following commutation relations
\begin{equation}
[\sum_{\beta=1}^rd_{\beta}^s\hat \Delta^s E_{\beta},\hat\Gamma^{(r)}_{\Delta}]=1,[\L_{\Delta},\M_{\Delta}]=1,
\end{equation}
which can be verified by a straightforward calculation.
By using the Sato equation, the iso-spectral flow of the $\M_{\Delta}$ operator
is given by
 \begin{equation}
\partial_{t_{k,\beta}}\M_{\Delta}=[(C_{\Delta k}^{\beta r})_+,\M_{\Delta}].
\end{equation}  More generally,
\begin{equation}
\partial_{t_{k,\beta}}(\M_{\Delta}^m\L_{\Delta}^l)=[(C_{\Delta k}^{\beta r})_+,\M_{\Delta}^m\L_{\Delta}^l].
\end{equation}

Based on the above preparation, the additional symmetry flows
of the discrete constrained multi-component KP hierarchy
are define by their actions on the dressing operator
 \begin{equation}
\partial_{l,m,\beta}S_{\Delta}=-(\M_{\Delta}^mC_{\Delta l}^{\beta r})_-\circ S_{\Delta},
\end{equation}
or equivalently  on the Lax operator
\begin{eqnarray}
\partial_{l,m,\beta}\L_{\Delta}&=&[-(\M_{\Delta}^mC_{\Delta l}^{\beta r})_-,\L_{\Delta}]
\label{Dadditionalofdkp},
\end{eqnarray} where $\partial_{l,m,\beta}$ denotes
the derivative with respect to an additional new
variable $t_{l,m,\beta}^*$. The more general actions of the additional symmetry flows
of the discrete constrained multi-component KP are given by
\begin{equation}\label{DML}
\dfrac{\partial \M_{\Delta}}{\partial_{m,n,\beta}}=-[(\M_{\Delta}^mC_{\Delta n}^{(\beta,r)})_-,\M_{\Delta}],\ \
 \partial_{l,m,\beta}\M_{\Delta}^n\L_{\Delta}^k=
 [-(\M_{\Delta}^mC_{\Delta l}^{\beta r})_-,\M_{\Delta}^n\L_{\Delta}^k].
\end{equation}

Later we can prove the additional flows $\dfrac{\partial }{\partial
t_{m,n,\beta}^*}$  commute with the flow $\dfrac{\partial
}{\partial t_{k,\gamma}}$,  but do not
commute with each other.  To this purpose, we need several lemmas and propositions as preparation firstly.

Similarly as differential case, for above local(means only contain non-negative powers on $\Delta$) difference operators $B_{\Delta n}^{\beta,r}$, we have the following lemma.
\begin{lemma}
  \label{Dlemm} $[B_{\Delta n}^{\beta,r},\phi\hat \Delta^{-1}\psi]_-=B_{\Delta n}^{\beta,r}(\phi)\hat \Delta^{-1}\psi-\phi\hat \Delta^{-1}B_{\Delta n}^{(\beta,r)*}(\psi)$.
\end{lemma}
\begin{proof}
 Firstly we consider a fundamental monomial: $\hat \Delta^n$ ($n\ge 1$). Then
  \begin{displaymath}
    [\hat \Delta^n,\phi\hat \Delta^{-1}\psi]_-=(\hat \Delta^n(\phi))\hat \Delta^{-1}\psi- (\phi\hat \Delta^{-1}\psi \hat \Delta^n)_-.
  \end{displaymath}
  Notice that the second term can be rewritten in the following way
  \begin{eqnarray*}
    (\phi\hat \Delta^{-1}\psi \hat \Delta^{n})_-
    =(\phi(\hat \Lambda^{-1}(\psi ))\hat \Delta^{n-1}
        -\phi\hat \Delta^{-1}(\hat \Delta\Lambda^{-1}(\psi))\hat \Delta^{n-1})_- \\
    =(\phi\hat \Delta^{-1}(-\hat \Delta\Lambda^{-1}(\psi))\hat \Delta^{n-1})_-
    =\cdots
    =\phi\hat \Delta^{-1}\left((-\hat \Delta\Lambda^{-1})^n
    (\psi)\right)=\phi\hat \Delta^{-1}(\hat \Delta^n)^*(\psi),
  \end{eqnarray*}
  then the lemma is proved.
\end{proof}

We can get some properties of the Lax operator in the following proposition.

By induction, the following important lemma can be derived
\begin{lemma}The Lax operator $\L_{\Delta}$ of discrete constrained $N$-component KP hierarchy satisfied the relation of
\begin{equation}\label{DLk}
(\L_\Delta^k)_-= \sum_{j=0}^{k-1}\L_\Delta^{k-j-1}(\phi)\hat \Delta^{-1}{(\L_\Delta^*)}^j(\psi),\ \ k\in \Z.
\end{equation}
where $\L_\Delta(\phi)=(\L_\Delta)_{+}(\phi) +  \phi\hat \Delta^{-1}(\psi\phi).$
\end{lemma}

Specifically, for example the case when $k=3$,
\begin{equation}\label{DL3}
(\L_\Delta^3)_-=[\L_\Delta^2(\phi)\hat \Delta^{-1}\psi+ \L_\Delta(\phi)\hat \Delta^{-1}\L_\Delta^*(\psi)+ \phi\hat \Delta^{-1}\L_\Delta^{2*}(\psi)].
\end{equation}
According to the definition  of the discrete constrained multi-component KP hierarchy,
the action of original
additional flows of the discrete constrained $N$-component KP hierarchy  is  expressed by
\begin{equation}\label{Dadditional}
(\partial_{1,k,\beta} \L_\Delta)_-=[(\M_{\Delta} C_{\Delta k}^{(\beta,r)})_+,\L_\Delta]_-+(C_{\Delta k}^{(\beta,r)})_-.
\end{equation}
Also  if $
X=M\hat \Delta^{-1}N,$ then
\begin{equation}\label{DoperatorXL}
[X,\L_\Delta]_-=[M\hat \Delta^{-1}\L_\Delta^*(N)-\L_\Delta(M)\hat \Delta^{-1}N]+ [X(\phi)\hat \Delta^{-1}\psi-\phi\hat \Delta^{-1}X^*(\psi)].
\end{equation}

To keep the consistency with flow equations on wave function $\phi,\psi$, we shall introduce an operator
$T_{\Delta k}^{(\alpha,r)}$  to modify the additional symmetry
of the discrete constrained multi-component KP hierarchy.
We now introduce a pseudo discrete operator  $T_{\Delta k}^{(\alpha,r)}$,
\begin{eqnarray}
T_{\Delta k}^{(\alpha,r)}&=&0,\ \ k=-1,0,1,\label{Dy12}\\
T_{\Delta k}^{(\alpha,r)}&=&\sum_{j=0}^{k-1}[j-\frac{1}{2}(k-1)]C_{\Delta k-1-j}^{(\alpha,r)}(\phi)
\hat \Delta^{-1}(C_{\Delta j}^{(\alpha,r)*})(\psi),k\geq 2,\label{Dy}
\end{eqnarray}
and have the following property.
\begin{lemma}
The action of flows $\partial_{t_{l,\beta}}$ of the discrete constrained $N$-component KP hierarchy  on the
$T_{\Delta k}^{(\alpha,r)}$ is
\begin{equation}\label{Dykderivat}
\partial_{t_{l,\beta}}
T_{\Delta k}^{(\alpha,r)}=[B_{\Delta l}^{(\beta,r)},T_{\Delta k}^{(\alpha,r)}]_-.
\end{equation}
\end{lemma}
\begin{proof}
The proof is similar as continuous constrained multi-component KP hierarchy instead of using the  Lemma \ref{Dlemm}.
\end{proof}

Further, the following expression of $[T_{\Delta k-1}^{(\beta,r)},L]_-$ is also
necessary to define the additional flows of the discrete constrained $N$-component KP hierarchy.
\begin{lemma}
The Lax operator $\L_\Delta$ of discrete constrained $N$-component KP hierarchy and $T_{\Delta k-1}^{(\beta,r)}$ has the following relation,
\begin{eqnarray}\notag
[T_{\Delta k-1}^{(\beta,r)},\L_\Delta]_-&=&-(C_{\Delta k}^{(\beta,r)})_-+\frac{k}{2}[\phi_\Delta^{(r)}\hat \Delta^{-1}(C_{\Delta k-1}^{(\beta,r)})^*(\psi_\Delta^{(r)})+C_{\Delta k-1}^{(\beta,r)}(\phi_\Delta^{(r)})\hat \Delta^{-1}\psi_\Delta^{(r)}\\
\label{Dykl}
&&+T_{\Delta k-1}^{(\beta,r)}(\phi_\Delta^{(r)})\hat \Delta^{-1}\psi_\Delta^{(r)}-\phi_\Delta^{(r)}\hat \Delta^{-1}(T_{\Delta k-1}^{(\beta,r)})^*(\psi_\Delta^{(r)})].
\end{eqnarray}
\end{lemma}
\begin{proof}
A similar direct calculation as the continuous case can  help us deriving eq.\eqref{Dykl}.
\end{proof}

We define the
additional flows of the discrete constrained $N$-component KP hierarchy  as
 \begin{equation}\label{Dtkflow}
\partial_{t_{1,k,\beta}^*}\L_\Delta=[-(\M_\Delta C_ {\Delta k}^{(\beta,r)})_-+T_{\Delta k-1}^{\beta,r},\L_\Delta],
\end{equation}
where
$T_{\Delta k-1}^{\beta,r}=0$, for $l=0,1,2$, such that the right-hand side of
(\ref{Dtkflow}) is in the form of derivation of Lax equation. Generally, one can also derive
\begin{equation}\label{DMLK}
\partial_{t_{1,k,\beta}^*}(\M_{\Delta}\L_\Delta^l)=[-(\M_\Delta C_ {\Delta k}^{(\beta,r)})_-+T_{\Delta k-1}^{\beta,r},\M_{\Delta}\L_\Delta^l].
\end{equation}

Now we calculate the action of the additional flows eq.\eqref{Dtkflow}
on the eigenfunction $\phi$ and $\psi$ of the discrete constrained $N$-component KP hierarchy.
\begin{theorem}\label{Dsymmetre}
The acting of additional flows of discrete constrained $N$-component KP hierarchy on the eigenfunction $\phi$ and $\psi$ are
\begin{equation}\label{DBAfunction}
\begin{split}
{\partial_{t_{1,k,\beta}^*}\phi}&=(\M_{\Delta}
C_{\Delta k}^{(\beta,r)})_+(\phi)+T_{\Delta k-1}^{(\beta,r)}(\phi)+\frac{k}{2} C_{\Delta k-1}^{(\beta,r)}(\phi),\\
 {\partial_{t_{1,k,\beta}^*}\psi}&=-(\M_{\Delta}
C_{\Delta k}^{(\beta,r)})^*_+(\psi)-(T_{\Delta k-1}^{(\beta,r)})^*(\psi)+\frac{k}{2}{ C_{\Delta k-1}^{(\beta,r)*}}(\psi).
\end{split}
\end{equation}

\end{theorem}

Next we shall prove the commutation relation between the additional
flows $\partial_{t_{1,k,\beta}^*}$ of discrete constrained $N$-component KP hierarchy and the original flows $\partial_{t_{l,\alpha}}$ of
it.
\begin{theorem}
The additional flows of $\partial_{t_{1,k,\beta}^*}$ are  symmetry flows of the discrete constrained $N$-component KP hierarchy, i.e. they commute with all $\partial_{t_{l,\beta}}$ flows.
\end{theorem}
\begin{proof}
According the action of  $\partial_{t_{1,k,\beta}^*}$ and $\partial_{t_{l,\alpha}}$ on the
dressing operator $S_{\Delta}$, then
\begin{eqnarray*}
[\partial_{t_{1,k,\beta}^*},\partial_{t_{l,\alpha}}]S_{\Delta} &=&[( C_{\Delta l}^{(\beta,r)})_+,-T_{\Delta k-1}^{(\beta,r)}]_-S_{\Delta}-(\partial_{t_{l,\alpha}} T_{\Delta k-1}^{(\beta,r)})S_{\Delta}\\
&=&0.
\end{eqnarray*}
Therefore the theorem holds.
\end{proof}

Then  using  the
relation
${{\partial_{t_{1,l,\beta}^*}}(\L_\Delta^k(\phi))}=({\partial_{t_{1,l,\beta}^*}}(\L_\Delta^k))(\phi)+\L_\Delta^k
{\partial_{t_{1,l,\beta}^*}}(\phi)$,
we can find
the additional  flows $\partial_{t_{1,l,\beta}^*}$ of discrete constrained $N$-component KP hierarchy have the following relations
\begin{eqnarray}\label{Dlqstar}
\begin{split}
{\partial_{t_{1,l,\beta}^*} \L_\Delta^k(\phi)}&=(\M_{\Delta}
C_{\Delta l}^{\beta,r})_+(\L_\Delta^k(\phi))+(k+\frac{l}{2})\L_\Delta^{k+l-1}(\phi),\ \ l=0,1,2,\\
{\partial_{t_{1,l,\beta}^*} {(\L_\Delta^*)}^k(\psi)}&=-(\M_{\Delta}
C_{\Delta l}^{\beta,r})^*_+{(\L_\Delta^*)}^k(\psi)+(k+\frac{l}{2}){(\L_\Delta^*)}^{k+l-1}(\psi),\ \ l=0,1,2.
\end{split}
\end{eqnarray}

Moreover, the action of
 $\partial_{t_{1,l,\beta}^*}$ on $T_{\Delta k}^{(\alpha,r)}$ is given by the following lemma.
\begin{lemma}
The actions on $T_{\Delta k}^{(\alpha,r)}$ of the additional  symmetry flows
$\partial_{t_{1,l,\beta}^*}$ of the discrete constrained $N$-component KP hierarchy are
\begin{eqnarray}\label{DPL}
\partial_{t_{1,l,\beta}^*} T_{\Delta k}^{(\alpha,r)}=[(\M_{\Delta} C_{\Delta l}^{\beta,r})_++T_{\Delta l-1}^{(\beta,r)},T_{\Delta k}^{(\alpha,r)}]_-+(k-l+1)T_{\Delta k+l-1}^{(\alpha,r)}\delta_{\alpha,\beta}.
\end{eqnarray}
\end{lemma}
\begin{proof} Using the same technique as continuous case, a straightforward calculation implies eq.\eqref{DPL}.
\end{proof}

Now it is time to identity the algebraic structure of the
additional symmetry flows of the discrete constrained $N$-component KP hierarchy.
\begin{theorem}\label{Dalg}
The additional flows $\partial_{t_{1,k,\beta}^*}$ of the discrete constrained $N$-component KP hierarchy form the
positive half of Virasoro algebra, i.e., for $l=0,1,2; k\geq 0,$
\begin{equation}
[\partial_{t_{1,l,\alpha}^*},\partial_{t_{1,k,\beta}^*}]=(k-l)\delta_{\alpha,\beta}\partial_{t_{1,k+l-1,\alpha}^*}.
\end{equation}
\end{theorem}
\begin{proof}
Using $[(\M_{\Delta}C_{\Delta l}^{\alpha,r})_-,T_{\Delta k-1}^{(\beta,r)}]_-=[(\M_{\Delta}C_{\Delta l}^{\alpha,r})_-,T_{\Delta k-1}^{(\beta,r)}]$
and the Jacobi identity, one can derive the following computation
\begin{eqnarray*}
&&[\partial_{t_{1,l,\alpha}^*},\partial_{t_{1,k,\beta}^*}]\L_\Delta\\
&=&\partial_{t_{1,l,\alpha}^*}([-(\M_\Delta C_{\Delta k}^{\beta,r})_-,\L_\Delta]+[T_{\Delta k-1}^{(\beta,r)},\L_\Delta])-\partial_{t_{1,k,\beta}^*}([-(\M_\Delta C_{\Delta l}^{\alpha,r})_-,\L_\Delta]) \\
&=&\partial_{t_{1,l,\alpha}^*} [-(\M_\Delta C_ {\Delta k}^{(\beta,r)})_-,\L_\Delta] +[\partial_{t_{1,l,\alpha}^*} T_{\Delta k-1}^{(\beta,r)},\L_\Delta]+[T_{\Delta k-1}^{(\beta,r)},\partial_{t_{1,l,\alpha}^*} \L_\Delta]+[\partial_{t_{1,k,\beta}^*}(\M_\Delta C_{\Delta l}^{\alpha,r})_-,\L_\Delta]\\
&&+[(\M_\Delta C_{\Delta l}^{\alpha,r})_-,\partial_{t_{1,k,\beta}^*} \L_\Delta]\\
&=&[-(\partial_{t_{1,l,\alpha}^*} (\M_\Delta C_ {\Delta k}^{(\beta,r)}))_-,\L_\Delta] +[-(\M_\Delta C_{\Delta k}^{\beta,r})_-,(\partial_{t_{1,l,\alpha}^*} \L_\Delta)] +[[(\M_{\Delta}C_{\Delta l}^{\alpha,r}R_{\Delta\beta})_+,T_{\Delta k-1}^{(\beta,r)}]_-\\
&&+(k-l)T_{k+l-2}^{(\beta,r)}\delta_{\alpha,\beta},\L_\Delta]+[T_{\Delta k-1}^{(\beta,r)},[-(\M_{\Delta}C_{\Delta l}^{\alpha,r})_-,\L_\Delta]]
+[[-(\M_\Delta C_ {\Delta k}^{(\beta,r)})_-\\
&&+T_{\Delta k-1}^{(\beta,r)},\M_{\Delta}C_{\Delta l}^{\alpha,r}]_-,\L_\Delta]+[(\M_\Delta C_{\Delta l}^{\alpha,r})_-,[-(\M_\Delta C_ {\Delta k}^{(\beta,r)})_-+T_{\Delta k-1}^{(\beta,r)},\L_\Delta]] \\
&=&(k-l)[-(\M_{\Delta}\L_\Delta^{k+l-1})_-,\L_\Delta]+(k-l)[\delta_{\alpha,\beta}T_{k+l-2},\L_\Delta]\\
&&-[[(\M_{\Delta}C_{\Delta l}^{\alpha,r})_-,T_{\Delta k-1}^{(\beta,r)}],\L_\Delta]
+[T_{\Delta k-1}^{(\beta,r)},[-(\M_{\Delta}C_{\Delta l}^{\alpha,r})_-,\L_\Delta]]
+[(\M_\Delta C_{\Delta l}^{\alpha,r})_-,[T_{\Delta k-1}^{(\beta,r)},\L_\Delta]]\\
&=&(k-l)\delta_{\alpha,\beta}\partial_{t_{1,k+l-1,\alpha}^*}\L_\Delta.
\end{eqnarray*}

\end{proof}
This shows  that
the discretization from constrained KP hierarchy to discrete constrained $N$-component KP hierarchy
keeps the important virasoro algebraic structure.

\section{Conclusions and Discussions}\label{Dconclusion}

 In this paper, the additional symmetry flows in eq.(\ref{tkflow}) for the constrained $N$-component KP hierarchy have been constructed. In this process, with the help of
 differential  operator $T_{k}^{(\alpha,r)}$ , the additional flows (\ref{tkflow}) on eigenfunctions and adjoint eigenfunctions
 of   constrained $N$-component KP hierarchy were obtained in Theorem
\ref{symmetre}. In Theorem \ref{alg}, these flows have been shown to
provide a hidden algebraic structure, i.e., the Virasoro
algebra. Also we construct the additional symmetries and prove a lot of similar results of the discrete
constrained multi-component KP  hierarchy and give the Virasoro flow equations on eigenfunctions and adjoint eigenfunctions.  This tells us  that
the discretization from constrained KP hierarchy to discrete constrained $N$-component KP hierarchy
retain the important virasoro algebraic
structure.

{\bf Acknowledgments} {\noindent \small  Chuanzhong Li  is supported by the National Natural Science Foundation of China under Grant No. 11201251, 11571192, the Zhejiang Provincial Natural Science Foundation under Grant No. LY15A010004, LY12A01007, the Natural Science Foundation of Ningbo under Grant No. 2015A610157. Jingsong He is supported by the National Natural Science Foundation of China under Grant No. 11271210, K. C. Wong Magna Fund in
Ningbo University. }


\end{document}